\documentclass[11pt]{article}

\newif\iffullver
\fullverfalse
\fullvertrue 

\newcommand{\short}[1]{\iffullver\else#1\fi}

\usepackage{amsmath, amssymb, amsthm}
\usepackage{algorithm, algorithmic}
\usepackage{xspace}

\usepackage{times}
\short{\usepackage[small,compact]{titlesec}}

\advance \topmargin by -\headheight
\advance \topmargin by -\headsep

\evensidemargin \oddsidemargin
\newcommand{\idc}{{iterative database construction}\xspace}

\newcommand{\IDC}{{Iterative Database Construction}\xspace}

\newcommand{\dus}{database update sequence\xspace}
\newcommand{\DUS}{Database Update Sequence\xspace}
\newcommand{\update}{\mathbf{U}}

\newcommand{\univ}{\mathcal{X}}
\newcommand{\cx}{\mathcal{X}}
\newcommand{\db}{\mathcal{D}}
\newcommand{\dba}{\mathcal{D}}
\newcommand{\dbb}{\mathcal{D}'}
\newcommand{\dbas}{\mathcal{D}^{\textrm{SMW}}}
\newcommand{\dbs}{\mathbb{N}^{|\univ|}}
\newcommand{\dbstruct}{\mathcal{R}_{\update}}

\newcommand{\query}{Q}
\newcommand{\queryset}{\mathcal{Q}}

\newcommand{\querystep}[1]{\query_{#1}}
\newcommand{\dbstep}[1]{\db_{#1}}

\newcommand{\maxupdates}{B}
\newcommand{\updates}{T}
\newcommand{\Lap}{\mathrm{Lap}}
\newcommand{\acc}{\alpha}

\newcommand{\real}{A}
\newcommand{\noisyreal}{\widehat{\real}}

\newcommand{\noisyrealstep}[1]{\noisyreal_{#1}}

\newcommand{\set}[1]{\left\{#1\right\}}

\newcommand{\eqdef}{\mathbin{\stackrel{\rm def}{=}}}

\newcommand{\Ex}{\mathbb{E}}

\newcommand{\R}{\mathbb{R}}

\newtheorem{theorem}{Theorem}[section]
\newtheorem{lemma}[theorem]{Lemma}

\newtheorem{claim}[theorem]{Claim}
\newtheorem{remark}[theorem]{Remark}
\newtheorem{corollary}[theorem]{Corollary}

\theoremstyle{definition}
\newtheorem{definition}[theorem]{Definition}

\newcommand{\initOneLiners}{%
    \setlength{\itemsep}{0pt}
    \setlength{\parsep }{0pt}
    \setlength{\topsep }{0pt}
}
\newenvironment{OneLiners}[1][\ensuremath{\bullet}]
    {\begin{list}
        {#1}
        {\initOneLiners}}
    {\end{list}}

\author{Avrim Blum\thanks{Department of Computer Science, Carnegie Mellon University, Pittsburgh
    PA 15213. Email: {\tt avrim@cs.cmu.edu}} \and Aaron Roth\thanks{Department of Computer and Information Science, University of Pennsylvania, Philadelphia PA 19104.  Email: {\tt aaroth@cis.upenn.edu}}}
\title{\bf Fast Private Data Release Algorithms for Sparse Queries}

\begin{document}
\maketitle
\sloppy

\begin{abstract}
We revisit the problem of accurately answering large classes of statistical queries while preserving differential privacy. Previous approaches to this problem have either been very general but have not had run-time polynomial in the size of the database, have applied only to very limited classes of queries, or have relaxed the notion of worst-case error guarantees. In this paper we consider the large class of \emph{sparse} queries, which take non-zero values on only polynomially many universe elements. We give efficient query release algorithms for this class, in both the interactive and the non-interactive setting. Our algorithms also achieve better accuracy bounds than previous general techniques do when applied to sparse queries: our bounds are independent of the universe size. In fact, even the runtime of our interactive mechanism is independent of the universe size, and so can be implemented in the ``infinite universe'' model in which no finite universe need be specified by the data curator.
\end{abstract}

\section{Introduction}
A database $\dba$ represents a finite collection of individual records from some \emph{data universe} $\univ$, which represents the set of all \emph{possible} records. We typically think of $\univ$ as being extremely large: exponentially large in the size of the database, or in some cases, possibly even infinite. A fundamental task in private data analysis is to accurately answer statistical queries about a database $\dba$, while provably preserving the privacy of the individuals whose records are contained in $\dba$. The privacy solution concept we use in this paper is \emph{differential privacy}, which has become standard, and which we define in section \ref{sec:prelims}.

Accurately answering statistical queries is the most well studied problem in differential privacy, and the results to date come in two types. There are a large number of extremely general and powerful techniques (see for example \cite{BLR08,DNRRV09,DRV10,RR10, HT10, HR10}) that can accurately answer arbitrary families of statistical queries which can be exponentially large in the size of the database. Unfortunately, these techniques all have running time that is at least linear in the size of the data universe $|\cx|$ (i.e. possibly \emph{exponential} in the size of the database), and so are in many cases impractical. There are also several techniques that do run in polynomial time, but that are limited: either they can answer queries from a very general and structurally rich class (i.e. all low-sensitivity queries), but can only answer a linear number of such queries (i.e. \cite{DMNS06}), or they can answer a very large number of queries, but only from a structurally very simple class (i.e. intervals on the unit line\footnote{The algorithm of \cite{BLR08} can be generalized to answer axis-aligned rectangle queries in constant dimension, but this is still a class that has only constant VC-dimension.} \cite{BLR08}), or as in several recent results (for conjunction and parity queries respectively) \cite{GHRU11,HRS11} they run in polynomial time, but offer only average case guarantees for randomly chosen queries. One of the main open questions in data privacy is to develop general data release techniques comparable in power to the known exponential time techniques that run in polynomial time. There is evidence, however, that this is not possible for arbitrary linear queries \cite{DNRRV09, UV11, GHRU11}.

In this paper, we consider a restricted but structurally rich class of linear queries which we call $\emph{sparse}$ queries. We say that a query is $m$-sparse if it takes non-zero values on only $m$ universe elements, and that a class of queries is $m$-sparse if each query it contains is $m'$ sparse for some $m' \leq m$. We will typically think of $m$ as being some polynomial in the database size $n$. Note that although each individual query is restricted to have support on only a polynomially sized subset of the data universe, different queries in the same class can have different supports, and so a class of sparse queries can still have support over the entire data universe. Note that the class of $m$-sparse queries is both very large (of size roughly $|\cx|^m$), and very structurally complex (the class of $m$-sparse queries have VC-dimension $m$). Sparse queries represent questions about individuals whose answer is rarely ``yes'' when asked about an individual who is drawn uniformly at random from the data population. Nevertheless, such questions can be useful to a data analyst who has some knowledge about which segment of the population a database might be drawn from. For example, a database resulting from a medical study might contain individuals who have some rare disease, but the data analyst does not know \emph{which} disease -- although there may be many such queries, each one is sparse. Alternately, a data analyst might have knowledge about the participants of several previous studies, and might want to know how much overlap there is between the participants of each previous study and of the current study. In general, sparse queries will only be useful to a data analyst who has some knowledge about the database, beyond that it is merely a subset of an exponentially sized data universe. Our results can therefore be viewed as a way of privately releasing information about a database that is useful to specialists -- but is privacy preserving no matter who makes use of it. In general, this work can be thought of as part of an agenda to find ways to make use of the \emph{domain knowledge} of the data analyst, to make private analysis of large-scale data-sets feasible.

\subsection{Results}
We give two algorithms for releasing accurate answers to $m$-sparse queries while preserving differential privacy: one in the interactive setting, in which the data curator acts as an intermediary and must answer an adaptively chosen stream of queries as they arrive, and one in the non-interactive setting, in which the data curator must in one shot output a data-structure which encodes the answers to every query of interest. In the interactive setting, we require that the running time needed to answer each query is bounded by a polynomial in $n$, the database size (so to answer any sequence of $k$ queries takes time $k\cdot \textrm{poly}(n)$). In the non-interactive setting, the entire computation must be performed in time polynomial in $n$, and the time required to evaluate any query on the output data structure must also be polynomial. Therefore, from the point of view of running time, the non-interactive setting is strictly more difficult than the interactive setting.

In the interactive setting, we give the following utility bound:
\begin{theorem}[Informal, some parameters hidden]
There exists an $(\epsilon,\delta)$-differentially private query release mechanism in the interactive setting, with running time per query $\tilde{O}(m/\alpha^2)$ that is $\alpha$-accurate with respect to any set of $k$ adaptively chosen $m$-sparse queries with:
$$\alpha = O\left(\frac{\left(\log m\right)^{1/4}\left(\log\frac{1}{\delta}\log k\right)^{1/2}}{(\epsilon n)^{1/2}}\right)$$
\end{theorem}

In the non-interactive setting, we give the bound:
\begin{theorem}[Informal, some parameters hidden]
There exists an $(\epsilon,\delta)$-differentially private query release mechanism in the non-interactive setting, with running time polynomial in the database size $n$, $m$, and $\log |\univ|$, that is $\alpha$-accurate with respect to any class of $k$ $m$-sparse linear queries, with:
$$\alpha = \tilde{O}\left(\log k\frac{\sqrt{m\log\left(\frac{1}{\delta}\right)}}{\epsilon n}\right)$$
\end{theorem}

Several aspects of these theorems are notable. First, the accuracy bounds do not have any dependence on the size of the data universe $|\univ|$, and instead depend only on the sparsity parameter $m$. Therefore, in addition to efficiency improvements, these results give accuracy improvements for sparse queries, when compared to the general purpose (inefficient) mechanisms for linear queries, which typically have accuracy which depends on $\log |\univ|$. Since we typically view $|\univ|$ as exponentially large in the database size, whereas $m$ is only polynomially large in the database size for these algorithms to be efficient, this can be a large improvement in accuracy.

Second, the interactive mechanism does not even have a dependence on $|\univ|$ in its running time! In fact, it works even in an \emph{infinite} universe (e.g. data entries with string valued attributes without pre-specified upper bound on length)\footnote{The algorithm must be able to read a \emph{name} for each universe element it deals with, and so it can of course not deal with elements that have no finite description length. But for a (countably) infinite universe, the running time would depend on the length of the largest string used to denote a universe element encountered during the running of the algorithm, and not in any a-priori way on the (unboundedly large) size of the universe.}. In this setting, queries may still be concisely specified as a list of polynomially many individuals from the possibly infinite universe that satisfy the query. Moreover, because the accuracy of this mechanism depends only very mildly on $m$, and the running time is linear in $m$, it can be used to answer $m$-sparse queries for arbitrarily large polynomial values of $m$, where the mechanism is constrained only by the available computational resources.

The non-interactive mechanism in contrast has a worse dependence on $m$. This bound essentially matches the error that would result from releasing the perturbed \emph{histogram} of the database, but does so in a way that requires computation and output representation only polynomial in $n$ (rather than linear in $|X|$, as releasing a histogram would require). Because accuracy bounds $> 1$ are trivial, this mechanism only guarantees non-trivial accuracy for $m$-sparse queries with $m << n^2/\log k$ (This is still of course a very large class of queries: there are roughly $|\univ|^{n^2/\log k}$ such queries, i.e., super-exponentially many in $n$). Nevertheless, there are distinct advantages to having a non-interactive mechanism that only needs to be run once. This is among the first \emph{polynomial time} non-interactive mechanisms for answering an exponentially large, unstructured class of queries while preserving differential privacy.

We note that our results give as a corollary, more efficient algorithms for answering conjunctions with many literals. This complements the beautiful recent work of Hardt, Rothblum, and Servedio \cite{HRS11}, who give more efficient algorithms for answering conjunctions with few literals, based on reductions to threshold learning problems.

\subsection{Techniques}
Our interactive mechanism is a modification of the very general multiplicative weights mechanism of Hardt and Rothblum \cite{HR10}. We give the interactive mechanism via the framework of \cite{GRU11} which efficiently maps objects called \emph{iterative database constructions} (defined in section \ref{sec:interactive}) into private query release mechanisms in the interactive setting. IDC algorithms are very similar to online learning algorithms in the mistake bound model, and we use this analogy to implement a version of the multiplicative weights IDC of Hardt and Rothblum \cite{HR10} analogously to how the Winnow algorithm is implemented in the \emph{infinite attribute model} of learning, defined by Blum \cite{Blu90}. The algorithm roughly works as follows: the multiplicative weights algorithm normally maintains a distribution over $|\univ|$ elements, one for each element in the data universe. It can be easily implemented in such a way so that when it is updated after a query $\query$ arrives, only those weights corresponding to elements in the support of the query $\query$ are updated: for an $m$-sparse query, this means it only need update $m$ positions. It also comes with a guarantee that it never needs to perform more than $\log |\univ|/\alpha^2$ updates before achieving error $\alpha$, and so at most $m\log |\univ|/\alpha^2$ elements ever need to be updated. The key insight is to pick a smaller universe, $\widehat{\univ}$, such that $\widehat{\univ} \geq m\log\widehat{\univ}/\alpha^2$, but \emph{not to commit to the identity of the elements in this universe} before running the algorithm, letting all elements be initially unassigned. The algorithm then maintains a hash table mapping elements of $\univ$ to elements of $\widehat{\univ}$. Elements in $\univ$ are assigned temporary mappings to elements in $\widehat{\univ}$ as queries come in, but are only assigned permanent mappings when an update is performed. Because only $\log\widehat{\univ}/\alpha^2$ updates are ever performed, and $\widehat{\univ}$ was chosen such that $\widehat{\univ} \geq m\log\widehat{\univ}/\alpha^2$, the algorithm never runs out of elements of $\widehat{\univ}$ to permanently assign. Because $|\widehat{\univ}|$ depends only on the desired accuracy $\alpha$ and the sparsity parameter $m$, and \emph{not} on $\univ$ in any way, the algorithm can be implemented and run without any knowledge of $\univ$ (even for infinite universes), and neither the running time nor the resulting accuracy depend on $|\univ|$.

The non-interactive mechanism releases a random projection of the database into polynomially many dimensions, together with the corresponding projection matrix. Queries are evaluated by computing their projection using the public projection matrix, and then taking the inner product of the projected query and the projected database. The difficulty comes because the projection matrix projects vectors from $|\univ|$-dimensional space to $\textrm{poly}(n)$ dimensional space, and so normally would take $|\univ|\textrm{poly}(n)$-many bits to represent. Our algorithms are constrained to run in time poly$(n)$, however, and so we need a concise representation of the projection matrix. We achieve this by using a matrix implicitly generated by a family of limited-independence hash functions which have concise representations. This requires using a limited independence version of the Johnson-Lindenstrauss lemma, and of concentration bounds. This algorithm also gives accuracy bounds which are independent of $|\univ|$.
\subsection{Related Work}

Differential privacy was introduced by Dwork, McSherry, Nissim, and Smith \cite{DMNS06}, and has since become the standard solution concept for privacy in the theoretical computer science literature. There is now a vast literature concerning differential privacy, so we mention here only the most relevant work, without attempting to be exhaustive. Dwork et al. \cite{DMNS06} also introduced the \emph{Laplace} mechanism, which is able to efficiently answer arbitrary low-sensitivity queries in the interactive setting. The Laplace mechanism does not make efficient use of the \emph{privacy budget} however, and can answer only linearly many queries in the database size.

Blum, Ligett, and Roth \cite{BLR08} showed that in the non-interactive setting, it is possible to answer \emph{exponentially} sized families of counting queries. This result was extended and improved by Dwork et al. \cite{DNRRV09} and Dwork, Rothblum, and Vadhan \cite{DRV10}, who gave improved running time and accuracy bounds, and for $(\epsilon,\delta)$-differential privacy gave similar results for arbitrary low sensitivity queries. Roth and Roughgarden \cite{RR10} showed that accuracy bounds comparable to \cite{BLR08} could be achieved even in the \emph{interactive} setting, and this result was improved in both accuracy and running time by Hardt and Rothblum, who give the multiplicative weights mechanism, which achieves nearly optimal accuracy and running time \cite{HR10}. Gupta, Roth, and Ullman \cite{GRU11} generalize the algorithms of \cite{RR10, HR10} into a generic framework in which objects called \emph{iterative database constructions} efficiently reduce to private data release mechanisms in the interactive setting. Unfortunately, the running time of all of the algorithms discussed here is at least linear in $|\univ|$, and so typically exponential in the size of the private database. Moreover, there are both computational and information theoretic lower bounds suggesting that it may be very difficult to give private release algorithms for generic linear queries with substantially better run time \cite{DNRRV09, UV11, GHRU11}. As in this work, these algorithms give a guarantee on the worst-case error of any answered query.

There is also a small body of work giving more efficient query release mechanisms for specific classes of queries. \cite{BLR08} gave an efficient (running time polynomial in the database size $n$) algorithm for releasing the answers for 1-dimensional intervals on the discretized unit-line in the non-interactive setting. As far as we know, prior to this work, this was the only efficient mechanism in either the interactive or non-interactive settings for releasing the answers to an exponentially sized family of queries with worst-case error. This class is however structurally very simple: it has VC-dimension only $2$. Other efficient algorithms relax the notion of utility, no longer guaranteeing worst-case error for all queries. \cite{BLR08} also give an efficient algorithm for releasing \emph{halfspace} queries in the unit sphere, but this algorithm only guaranteed accurate answers for halfspaces that happened to have large \emph{margin} with respect to the points in the database. Gupta et al \cite{GHRU11} gave an algorithm for releasing \emph{conjunctions} over $d$ attributes to \emph{average} error $\alpha$ over any product distribution (over conjunctions), which runs in time $d^{O(1/\alpha)}$. This was improved to have running time $O(d^{\log 1/\alpha})$ by Cheraghchi et al. \cite{CKKL11}. Note that these algorithms only run in polynomial time for constant values of $\alpha$, and only give accuracy bounds in expectation over random queries. Recently, Hardt, Rothblum, and Servedio \cite{HRS11} gave an algorithm for releasing conjunctions defined on $k$ out of $d$ literals with an average-error guarantee \emph{for any} pre-specified distribution in time $d^{\tilde{O}(\sqrt{k})}$. Using the private boosting algorithm of \cite{DRV10}, they leverage this result to give an algorithm for releasing $k$-literal conjunctions with worst-case error guarantees, which increases the running time to $d^{\tilde{O}(k)}$, although still only requiring databases of size $d^{\tilde{O}(\sqrt{k})}$. They also gave an efficient (i.e. running time polynomial in $n$) algorithm for releasing \emph{parity} queries to low average error over product distributions. We remark that our results give a complementary bound for large conjunctions (with a better sample complexity requirement). Our online algorithm can release all  conjunctions on $d-k$ out of $d$ literals with worst-case error guarantees in time $d^{\tilde{O}(k)}$, requiring databases of size only $\tilde{O}(k^{1.5}\log d)$.

The efficient interactive mechanism we give in section \ref{sec:interactive} is based on an analogy between iterative database construction (IDC) algorithms and online learning algorithms in the mistake bound model. We implement the multiplicative weights IDC of Hardt and Rothblum \cite{HR10} analogously to how Winnow is implemented in the \emph{infinite attribute model} of Blum \cite{Blu90}. In our setting, it can be thought of as an \emph{infinite universe model} that has no dependence on the universe size in either the running time or accuracy bounds. This involves running the multiplicative weights algorithm on a much smaller universe. Hardt and Rothblum \cite{HR10} also gave a version of their algorithm which ran on a small subset of the universe to give efficient run-time guarantees. The main difference is that we select the subset of the universe that we run the multiplicative weights algorithm on adaptively, based on the queries that arrive, whereas \cite{HR10} select the subset nonadaptively, independently of the queries. \cite{HR10} give average case utility bounds for linear queries on randomly selected databases; in contrast, we give worst-case utility bounds that hold for all input databases, but only for sparse linear queries.

The efficient non-interactive mechanism we give in section \ref{sec:noninteractive} is based on random projections using families of limited independence hash functions, which have previously been used for space-bounded computations in the streaming model \cite{CW09,KN11}. Limited independence hash functions have also previously been used for streaming algorithms in the context of differential privacy \cite{DNPRY10}.
\section{Preliminaries}
\label{sec:prelims}
A database $\dba$ is a multiset of elements from some (possibly infinite) abstract universe $\univ$. We write $|\dba| = n$ to denote the cardinality of $\dba$. For any $x \in \univ$ we can also write $D[x]$ to denote: $\dba[x] = \{x' \in \dba : x' = x\}$ the number of elements of type $x$ in the database. Viewed this way, a database $\dba \in \mathbb{N}^{|\univ|}$ is a vector with integer entries in the range $[0,n]$.

A linear query $\query:\univ\rightarrow [0,1]$ is a function mapping elements in the universe to values on the real unit interval. For notational convenience, we will define $\query(\emptyset) = 0$. We can also evaluate a linear query on a database. The value of a linear query $\query$ on a database is simply the average value of $\query$ on elements of the database:
$$\query(\dba) = \frac{1}{n}\sum_{x \in \dba}\query(x) = \frac{1}{n}\sum_{x \in \univ}\query(x)D[x]$$
Similarly to how we can think of a database as a vector, we can think of a query as a vector $\query \in [0,1]^{|\univ|}$ with $\query[x] = \query(x)$. Viewed this way, $\query(\dba) = \frac{1}{n} \langle \query, \dba \rangle$.

It will sometimes be convenient to think of normalized databases (with entries that sum to 1). For a database $\dba$ of size $n$, we define the corresponding normalized database $\hat{\dba}$ to be the database such that $\hat{\dba}[x] = \dba[x]/n$. We evaluate a linear query on a normalized database by computing $\query(\hat{\dba}) = \sum_{x \in \univ}\query(x)\hat{\dba}[x] = \langle \query, \hat{\dba}\rangle$. Note that $\query(\dba) = \query(\hat{\dba})$.

\begin{definition}[Sparsity]The \emph{sparsity} of a linear query $\query$ is $|\{x \in \univ : \query(x) > 0\}|$, the number of elements in the universe on which it takes a non-zero value. We say that a query is $m$-sparse if its sparsity is at most $m$. We will also refer to the class of all $m$-sparse linear queries, denoted $\queryset_m$.
\end{definition}
In this paper, we will assume that given an $m$-sparse query, we can quickly (in time polynomial in $m$) enumerate the elements $x \in \univ$ on which $\query(x) > 0$.

\begin{remark}
While the assumption that we can quickly enumerate the non-zero values of a query may not always hold, it is indeed the case that for many natural classes of queries, we can enumerate the non-zero elements in time \emph{linear} in $m$. For example, this holds for queries that are specified as lists of the universe elements on which the query is non-zero, as well as for many implicitly defined query classes such as conjunctions, disjunctions, parities, etc.\footnote{The set of conjunctions over the $d$-dimensional boolean hypercube with $d - log(n)$ literals are $n$-sparse. Even though there are superpolynomially many such conjunctions, it is simple to enumerate the entries on which these conjunctions take non-zero value in time linear in $n$. We can simply enumerate all of the $2^{\log n} = n$ values that the unassigned variables can take.} Of course, classes like conjunctions are typically not sparse, but conjunctions with $d - O(\log n)$ literals are, and their support can be quickly enumerated (even though there are superpolynomially many such conjunctions).
\end{remark}

\subsection{Utility}
We will design algorithms which can accurately answer large numbers of sparse linear queries. We will be interested in both \emph{interactive} mechanisms and \emph{non-interactive} mechanisms. A non-interactive mechanism takes as input a database, runs one time, and outputs some data structure capable of answering many queries without further interaction with the data release mechanism. An interactive mechanism takes as input a stream of queries, and must provide a numeric answer to each query before the next one arrives.

\begin{definition}[Accuracy for non-Interactive Mechanisms]
  Let $\queryset$ be a set of queries.
  A non-interactive mechanism $M:  \univ^* \to R$ for some abstract range $R$ is \emph{$(\alpha,\beta)$-accurate for $\queryset$}
  if there exists a function $\mathrm{Eval}: \queryset \times R \to \mathbb{R}$ s.t.
  for every database $\dba \in \univ^*$, with
  probability at least $1-\beta$ over the coins of $M$, $M(\dba)$ outputs $r \in R$
  such that $\max_{\query \in \queryset}|\query(\dba)- \mathrm{Eval}(\query, r)| \leq \alpha$.
  We will abuse notation and write $\query(r) = \mathrm{Eval}(\query, r)$.

  $M$ is \emph{efficient} if both $M$ and $\mathrm{Eval}$ run in time polynomial in the size of the database $n$.
\end{definition}

\begin{definition}[Accuracy for Interactive Mechanisms]
Let $\queryset$ be a set of queries. An interactive mechanism $M$ takes as input an adaptively chosen stream of queries $\query_1,\ldots,\query_k \in \queryset$ and for each query $\query_i$, outputs an answer $a_i \in \mathbb{R}$ before receiving $\query_{i+1}$. It is $(\alpha,\beta)$-\emph{accurate} if for every database $\dba \in \univ^*$, with probability at least $1-\beta$ over the coins of $M$: $\max_{i}|\query_i - a_i| \leq \alpha$.

$M$ is \emph{efficient} if the update time for each query (i.e. the time to produce answer $a_i$ after receiving query $\query_i$) is polynomial in the size of the database $n$.
\end{definition}

\subsection{Differential Privacy}

We will require that our algorithms satisfy \emph{differential privacy}, defined as follows. We must first define the notion of \emph{neighboring databases}.
\begin{definition}[Neighboring Databases]
Two databases $\dba,\dbb$ are \emph{neighbors} if they differ only in the data of a single individual: i.e. if their symmetric difference is $|\dba \triangle \dbb| \leq 1$.
\end{definition}

\begin{definition}[Differential Privacy \cite{DMNS06}]
A randomized algorithm $M$ acting on databases and outputting elements from some abstract range $R$ is $(\epsilon,\delta)$-differentially private if for all pairs of neighboring databases $\dba,\dbb$ and for all subsets of the range $S \subseteq R$ the following holds:
$$\Pr[M(\dba) \in S] \leq \exp(\epsilon)\Pr[M(\dbb) \in S] + \delta$$
\end{definition}

\begin{remark}
For a non-interactive mechanism, $R$ is simply the set of data-structures that the mechanism outputs. For an interactive mechanism, because the queries may be adaptively chosen by an adversary, $R$ is the set of query/answer transcripts produced by the algorithm when interacting with an arbitrary adversary. For a detailed treatment of differential privacy and adaptive adversaries, see \cite{DRV10}.
\end{remark}

A useful distribution is the \emph{Laplace} distribution.
\begin{definition}[The Laplace Distribution]
The Laplace Distribution (centered at 0) with scale $b$ is the
distribution with probability density function: $\textstyle \Lap(x | b)
= \frac{1}{2b}\exp( -\frac{|x|}{b})$.
We will sometimes write $\textrm{Lap}(b)$ to denote the Laplace distribution with scale $b$, and will sometimes abuse notation and write $\Lap(b)$ simply to denote a random variable $X \sim \Lap(b)$.
\end{definition}
A fundamental result in data privacy is that perturbing low sensitivity queries with Laplace noise preserves $(\epsilon,0)$-differential privacy.
\begin{theorem}[\cite{DMNS06}]
\label{thm:laplace-privacy}
Suppose $\query:\univ^*\rightarrow \mathbb{R}$ is a function such that for all
neighboring databases $\db$ and $\db'$, $|\query(\db) - \query(\db')|
\leq c$. Then the procedure which on input $\db$ releases $\query(\db) + X$, where $X$ is a draw from a
$\textrm{Lap}(c/\epsilon)$ distribution, preserves
$(\epsilon,0)$-differential privacy.
\end{theorem}
It will be useful to understand how privacy
parameters for individual steps of an algorithm compose into privacy
guarantees for the entire algorithm.  The following useful theorem is a special case of a theorem proven by
Dwork, Rothblum, and Vadhan:
\begin{theorem}[Privacy Composition \cite{DRV10}]\label{thm:compose}
Let $0 < \epsilon, \delta < 1$, and let $M_1,\ldots,M_T$ be $(\epsilon',0)$-differentially private algorithms for some $\epsilon'$ at most:
$$\epsilon' \leq \frac{\epsilon}{\sqrt{8 T\log\left(\frac{1}{\delta}\right)}}.$$
Then the algorithm $M$ which outputs $M(\dba) = (M_1(\dba), \ldots, M_T(\dba))$ is $(\epsilon, \delta)$-differentially private.

\end{theorem} 

\section{A Fast IDC Algorithm For Sparse Queries}
\label{sec:interactive}
In this section we use the abstraction of an \emph{iterative database construction} that was introduced by Gupta, Roth, and Ullman \cite{GRU11}. It was shown in \cite{GRU11} that efficient IDC algorithms automatically reduce to efficient differentially private query release mechanisms in the interactive setting. Roughly, an IDC mechanism works by maintaining a sequence of data structures $\dbstep{1}, \dbstep{2}, \dots$ that give increasingly good approximations to the input database $\db$ (in a sense that depends on the IDC).  Moreover, these mechanisms produce the next data structure in the sequence by considering only one query $\query$ that \emph{distinguishes} the real database in the sense that $\query(\dbstep{t})$ differs significantly from $\query(\db)$.

Syntactically, we will consider functions of the form $\update: \dbstruct \times \queryset \times \R \to \dbstruct$.  The inputs to $\update$ are a data structure in $\dbstruct$, which represents the current data structure $\dbstep{t}$; a query $\query$, which represents the distinguishing query, and may be restricted to a certain set $\queryset$; and also a real number which estimates $\query(\db)$.  Formally, we define a \emph{\dus}, to capture the sequence of inputs to $\update$ used to generate the database sequence $\dbstep{1}, \dbstep{2}, \dots$.

\begin{definition}[\DUS]~\label{def:dus}
Let $\dba \in \dbs$ be any database and let \\ $\set{(\dbstep{t},
  \querystep{t}, \noisyrealstep{t})}_{t=1,\dots, \updates} \in
(\dbstruct \times \queryset \times \R)^{\updates}$ be a sequence of
tuples.   We say the sequence is an \emph{$(\update, \db, \queryset,
  \acc, \updates)$-\dus} if it satisfies the following properties:
\medskip
\begin{OneLiners}
\item[1.] $\dbstep{1} = \update(\emptyset, \cdot, \cdot)$,
\item[2.] for every $t = 1,2,\dots,\updates$, $\left| \querystep{t}(\db) - \querystep{t}(\dbstep{t}) \right| \geq \acc$,
\item[3.] for every $t = 1,2,\dots,\updates$, $\left| \querystep{t}(\db) - \noisyrealstep{t} \right| < \acc$,
\item[4.] and for every $t = 1,2, \dots, \updates-1$, $\dbstep{t+1} = \update(\dbstep{t}, \querystep{t}, \noisyrealstep{t})$.
\end{OneLiners}
\end{definition}

\begin{definition}[\IDC]~\label{def:idc}
Let $\update: \dbstruct \times \queryset \times \R \to \dbstruct$ be an update rule and let $\maxupdates: \R \to \R$ be a function.  We say $\update$ is a \emph{$\maxupdates(\acc)$-\idc for query class $\queryset$} if for every database $\db \in \dbs$, every $(\update, \db, \queryset, \acc, \updates)$-\dus satisfies $\updates \leq \maxupdates(\acc)$.
\end{definition}

Note that the definition of an $\maxupdates(\acc)$-\idc implies that if $\update$ is a $\maxupdates(\acc)$-\idc, then given any  maximal $(\update, \db, \queryset, \acc, \updates)$-\dus, the final database $\dbstep{\updates}$ must satisfy
$\max_{\query \in \queryset} \left| \query(\db) - \query(\dbstep{\updates}) \right| \leq \acc$
or else there would exist another query satisfying property 2 of Definition~\ref{def:dus}, and thus there would exist a $(\update, \db, \queryset, \acc, \updates+1)$-\dus, contradicting maximality.

$B(\alpha)$-IDC algorithms generically reduce to $(\epsilon,\delta)$-differentially private $(\alpha,\beta)$-accurate query release mechanisms in an efficiency preserving way. This framework was implicitly used by \cite{RR10} and \cite{HR10}.
\begin{theorem}[\cite{GRU11}]
\label{thm:IDCtoRelease}
If there exists a $B(\alpha)$-IDC algorithm for a class of queries $\queryset$ using a class of  datastructures $\dbstruct$  that take time at most $p(n,\alpha,|\cx|)$ to update their hypotheses, and time at most $q(n,\alpha,|\cx|)$ to evaluate a query on any $\dba \in \dbstruct$, then for any $0 < \epsilon,\delta,\beta < 1$ there exists an $(\epsilon,\delta)$-differentially private query release mechanism in the interactive setting that has update time at most $O(p(n,\alpha,\cx) + q(n,\alpha,\cx))$ and is $(\alpha,\beta)$-accurate with respect to any adaptively chosen sequence of $k$ queries from $\queryset$ where $\alpha$ is the solution to the following equality:
$$\alpha = \frac{3000\sqrt{B(\alpha)}\log(4/\delta)\log(k/\beta)}{\epsilon n}$$
\end{theorem}

In this section we will give an efficient IDC algorithm for the class of $m$-sparse queries, and then call on Theorem \ref{thm:IDCtoRelease} to reduce it to a differentially private query release mechanism in the interactive setting.

First we introduce the Sparse Multiplicative Weights data structure, which will be the class of datastructures $\dbstruct$ that the Sparse Multiplicative Weights IDC algorithm uses.:
\begin{definition}[Sparse Multiplicative Weights Data Structure]
The sparse multiplicative weights data structure $\dbas$ of size $s$ is composed of three parts. We write $\dbas = (\dba, h, \textrm{ind})$.
\begin{enumerate}
\item $\dba$ is a collection of $s$ real valued variables $x_1,\ldots,x_s$, with $x_i \in [0,1]$ for all $i \in [s]$. Variable $x_i$ for $i \in [s]$ is referenced by $\dba[i]$. Initially $x_i = 1/s$ for all $i \in [s]$. We define $\dba[i] = 0$ for all $i > s$.
\item $h$ is a hash table $h:\cx\rightarrow [s] \cup \emptyset$ mapping elements in the universe $X$ to indices $i \in [s]$. Elements $x \in \cx$ can also be unassigned in which case we write $h(x) = \emptyset$. Initially, $h(x) = \emptyset$ for all $x \in \cx$ We write $h^{-1}(i) = x$ if $h(x) = i$, and $h^{-1}(i) = \emptyset$ if there does not exist any $x \in \cx$ such that $h(x) = i$.
\item $\textrm{ind} \in [s+1]$ is a counter denoting the index of the first unassigned variable. For all $i < \textrm{ind}$, there exists some $x \in \cx$ such that $h(x) = i$. For all $i \geq \textrm{ind}$, there does not exist any $x \in \cx$ such that $h(x) = i$. Initially $\textrm{ind} = 1$.
\end{enumerate}

If $\textrm{ind} \leq s$, we can \emph{add} an unassigned element $x \in \cx$ to $\dbas$. Adding an element $x \in \cx$ to $\dbas$ sets $h(x) \leftarrow \textrm{ind}$ and increments $\textrm{ind} \leftarrow \textrm{ind} + 1$. If $\textrm{ind} = s+1$, attempting to add an element causes the data structure to report \textbf{FAILURE}.

A linear query $\query$ is evaluated on a sparse MW data structure $\dbas = (\dba,h)$ as follows.
$$\query(\dbas) = \sum_{x \in \cx : \query(x) > 0 \wedge h(x) \neq \emptyset}\query(x)\cdot \dba[h(x)] + \sum_{x \in \cx : \query(x) > 0 \wedge h(x) = \emptyset}\query(x)\cdot \dba[\textrm{ind}]$$
\end{definition}
We now present Algorithm \ref{alg:SMW}, the Sparse Multiplicative Weights (SMW) IDC algorithm for $m$-sparse queries. The algorithm is a version of the Hardt/Rothblum Multiplicative Weights IDC \cite{HR10}, modified to work without any dependence on the universe size. It will run multiplicative weights update steps over the variables of the SMW data structure, using the SMW data structure to delay assigning variables to particular universe elements $x \in \cx$ until necessary. Note that it is not simply running the multiplicative weights algorithm from \cite{HR10} implicitly: doing so would yield guarantees that depend on the cardinality of the universe $|\univ|$. Instead, the guarantees we will get will depend only on $m$, and so will carry over even to the infinite-universe setting.
\begin{algorithm}
\caption{The Sparse Multiplicative Weights (SMW) IDC Algorithm for $m$-sparse queries. It is instantiated with an accuracy parameter $\eta = \alpha/2$. It takes as input a sparse MW datastructure $\dbas$, an $m$-sparse query $\query \in \queryset_m$, and an estimate of the query value $\noisyreal$.}
\label{alg:SMW}
\textbf{SMW}($\dbas_t = (\dba_t,h_t, \textrm{ind}_t),\query_t,\noisyreal_t$):
\begin{algorithmic}
\IF{$\dbas_t = \emptyset$}
    \STATE \textbf{Let} $s$ be the smallest integer such that $s/(\log(s)+1) \geq 4 m/\alpha^2$.
    \STATE \textbf{Return} a new Sparse MW data structure $\dbas_1 = (D_1,h_1, \textrm{ind}_1)$ of size $s$ with $h_1(x) = \emptyset$ for all $x \in \cx$, $x_i = 1/s$ for all $i \in [s]$, and $\textrm{ind}_1 = 1$.
\ENDIF
\STATE \textbf{Let} $\dbas_{t+1} = (\dba_{t+1},h_{t+1},\textrm{ind}_{t+1}) \leftarrow \dbas_t$
\STATE \textbf{Update:} For all $x \in \cx$ such that $\query_t(x) > 0$: \textbf{If} $h_{t+1}(x)= \emptyset$ then \textbf{add} $x$ to $\dbas_{t+1}$.
\IF{$\noisyreal_t  <  \query_t(\dbas_t)$}
    \STATE \textbf{Update:} For all $x \in \cx$ such that $\query_t(x) > 0$: Let $$\dba_{t+1}[h_{t+1}(x)] \leftarrow \dba_{t+1}[h_{t+1}(x)]\cdot \exp(-\eta \query_t(x))$$
\ELSE
    \STATE \textbf{Update:} For all $x \in \cx$ such that $\query_t(x) > 0$: Let $$\dba_{t+1}[h_{t+1}(x)] \leftarrow \dba_{t+1}[h_{t+1}(x)]\cdot \exp(\eta \query_t(x))$$
\ENDIF
\STATE \textbf{Normalize:} For all $i \in [s]$: $$\dba_{t+1}[i] = \frac{\dba_{t+1}[i]}{\sum_{j=1}^{s}\dba_{t+1}[j]}$$
\STATE \textbf{Output} $\dbas_{t+1}$.
\end{algorithmic}
\end{algorithm}
\begin{theorem}
\label{thm:IDC}
The Sparse Multiplicative Weights algorithm is a $B(\alpha)$-IDC for the class of $m$-sparse queries $\queryset_m$, where:
$$B(\alpha) = 4\frac{\log s+1}{\alpha^2}$$
and $s$ is the smallest integer such that $s/(\log(s)+1) \geq 4 m/\alpha^2$.
\end{theorem}
The analysis largely follows the Multiplicative Weights analysis given by Hardt and Rothblum \cite{HR10}. The main difference is that rather than using one global potential function, we must use a different potential function for each database update sequence, defined as a function of the state of the hash table in the last SMW datastructure in the sequence. We must also argue that we never run out of variables to assign in the SMW data structure, which would cause it to return \textbf{FAILURE}. To argue this, we apply the technique of Blum \cite{Blu90}, used to adapt Winnow to the infinite attribute model.
\begin{proof}
We will consider any maximal $(\textrm{SMW}, \dbas, \queryset, \acc, \updates)$-database update sequence $\set{(\dbas_t,
  \querystep{t}, \noisyrealstep{t})}_{t=1,\dots, \updates}$. We will argue that $\updates \leq \frac{4\log s}{\alpha^2}$ and that no data structure $\dbas_t$ in the sequence ever returns \textbf{FAILURE} when the SMW algorithm attempts to \textbf{add} some element $x \in X$ to it. Consider the real private database $\dba$ and the final data structure in the sequence $\dbas_T = (\dba_T,h_T,\textrm{ind}_T)$. We will define a non-negative potential function $\Psi$ based on $h_T$ and $\hat{\dba}$ and show that it decreases significantly at each step. We define:
  $$\Psi_t \eqdef  \sum_{x : h_T(x) \neq \emptyset} \hat{\dba}[x]\log\left(\frac{\hat{\dba}[x]}{\dba_t[h_T(x)]}\right)$$

  \begin{claim}
  \label{claim:initial}
  For all $t \in [\updates]$, $\Psi_t \geq -\frac{1}{e}$ and $\Psi_0 \leq \log s$
  \end{claim}
  \begin{proof}
The log-sum inequality states that for any collection of non-negative numbers $a_1,\ldots,a_n$ and $b_1,\ldots,b_n$: $$\sum_{i=1}^na_i\log\left(\frac{a_i}{b_i}\right) \geq a\log\left(\frac{a}{b}\right)$$
 where $a = \sum_{i=1}^na_i$ and $b = \sum_{i=1}^n b_i$. We therefore have:
 \begin{eqnarray*}
 \Psi_t &=&  \sum_{x : h_T(x) \neq \emptyset} \hat{\dba}[x]\log\left(\frac{\hat{\dba}[x]}{\dba_t[h_T(x)]}\right)\\
  &\geq&  \left(\sum_{x : h_T(x) \neq \emptyset} \hat{\dba}[x]\right) \log\left(\frac{\sum_{x : h_T(x) \neq \emptyset} \hat{\dba}[x]}{\sum_{x : h_T(x) \neq \emptyset} \dba_t[h_T(x)]}\right) \\
  &\geq&\left(\sum_{x : h_T(x) \neq \emptyset} \hat{\dba}[x]\right) \log\left(\sum_{x : h_T(x) \neq \emptyset} \hat{\dba}[x]\right) \\
  &\geq& -\frac{1}{e}
 \end{eqnarray*}
 where the first inequality follows from the log-sum inequality, the second follows from the fact that $\sum_{x : h_T(x) \neq \emptyset} \dba_t[h_T(x)] \leq 1$, and the third follows from the fact that $\min_{a \in [0,1]}a\log a = -\frac{1}{e}$.
 To see that $\Psi_0 \leq \log s$, recall that $\dba_0[i] = 1/s$ for all $i$. Therefore:
 $$\Psi_0 =  \sum_{x : h_T(x) \neq \emptyset} \hat{\dba}[x]\log\left(s\hat{\dba}[x]\right)$$
 Since $\hat{\dba}$ is a probability distribution, this expression takes maximum value $\log s$.
  \end{proof}

  We will argue that in every step the potential drops by at least $\alpha^2/4$. Because the potential begins at $\log s$, and must always be non-negative, we therefore know that there can be at most $\updates \leq 4\log s/\alpha^2$ steps. To begin, let us see exactly how much the potential drops at each step:
\begin{lemma}
\label{lem:MWPotentialDrop}
$$\Psi_{t}-\Psi_{t+1} \geq \alpha^2/4$$
\end{lemma}
\begin{proof}
We follow the analysis of \cite{HR10}. We consider the case in which $\noisyreal_t  <  \query_t(\dbas_t)$. In this case:
\begin{eqnarray*}
\Psi_{t}-\Psi_{t+1} &=&  \sum_{x : h_T(x) \neq \emptyset} \hat{\dba}[x]\log\left(\frac{\hat{\dba}[x]}{\dba_t[h_T(x)]}\right) -  \sum_{x : h_T(x) \neq \emptyset} \hat{\dba}[x]\log\left(\frac{\hat{\dba}[x]}{\dba_{t+1}[h_T(x)]}\right) \\
&=& \sum_{x : h_T(x) \neq \emptyset} \hat{\dba}[x]\log\left(\frac{\dba_{t+1}[h_T(x)]}{\dba_{t}[h_T(x)]}\right) \\
&\geq& \sum_{x : h_T(x) \neq \emptyset} \hat{\dba}[x]\log\left(\frac{\exp(-\eta \query_t(x))\cdot \dba_{t}[h_T(x)]}{\dba_{t}[h_T(x)]}\right) - \log \left(\sum_{j =1}^s \exp(-\eta \query_t(h_t^{-1}(j)))\dba_{t}[j]\right) \\
&=& \sum_{x : \query_t(x) > 0} -\hat{\dba}[x] \eta \query_t(x) - \log \left(\sum_{j =1}^s \exp(-\eta \query_t(h_t^{-1}(j)))\dba_{t}[j]\right) \\
&=& -\eta \query_t(\dba) - \log \left(\sum_{j =1}^s \exp(-\eta \query_t(h_t^{-1}(j)))\dba_{t}[j]\right) \\
&\geq& -\eta \query_t(\dba) - \log \left(\sum_{j =1}^s (1-\eta\query_t(h_t^{-1}(j))+\eta^2)\dba_{t}[j]\right) \\
&=&-\eta \query_t(\dba) - \log \left(1+\eta^2 - \eta\sum_{x : \query_t(x) > 0}\query_t(x)\dba_{t}[h_t(x)]\right) \\
&\geq& \eta(\query_t(\dbas_t)-\query_t(\dba)) - \eta^2 \\
&\geq&\alpha^2/2 - \alpha^2/4 \\
&=& \alpha^2/4
\end{eqnarray*}
In this calculation, we used the facts that:
$$\exp(-\eta \query_t(x_i)) \leq 1 - \eta \query_t(x_i) + \eta^2\query_t(x_i)^2 \leq 1 - \eta \query_t(x_i) + \eta^2$$
that $\sum_{j=1}^s \dba_t[j] = 1$, that $\log(1+y) \leq y$ for $y > -1$, that by the definition of a database update sequence, when $\noisyreal_t  <  \query_t(\dbas_t)$ we also have that $\query_t(\dba) < \query_t(\dbas_t)$, and that by the definition of database update sequence we always have $|\query_t(\dbas_t)-\query_t(\dba)| \geq \alpha$. Finally we recall that $\eta = \alpha/2$ The case when $\noisyreal_t  >  \query_t(\dbas_t)$ is exactly similar.
\end{proof}
Theorem \ref{thm:IDC} then immediately follows by combining Claim \ref{claim:initial} with Lemma \ref{lem:MWPotentialDrop}:
$$-\frac{1}{e} \leq \Psi_T \leq \log s - \updates \cdot \frac{\alpha^2}{4}$$
Solving for $\updates$ we find:
$$\updates \leq 4 \frac{\log s+1/e}{\alpha^2} <  4 \frac{\log s+1}{\alpha^2}$$
Finally to see that the SMW data structure never reports \textbf{FAILURE}, it suffices to observe that $\textrm{ind}_T \leq s$. Because each query $\query_t$ is assumed to be $m$-sparse, at most $m$ variables can be \textbf{add}ed to the SMW data structure at each update. Therefore, we have
$$\textrm{ind}_T \leq m \cdot T \leq \frac{4m(\log s+1)}{\alpha^2}\leq s$$
The last inequality follows from recalling that we chose $s$ such that $s/(\log s+1) \geq 4m/\alpha^2$. This completes the proof.
\end{proof}

Finally, we may observe that both the update time for the SMW IDC and the time to evaluate a query on the SMW datatructure is $O(s) = \tilde{O}(m/\alpha^2)$. Therefore, we may instantiate Theorem \ref{thm:IDCtoRelease} with the SMW IDC algorithm to obtain the main result of this section:
\begin{theorem}
For any $0 < \epsilon,\delta,\beta < 1$ There exists an $(\epsilon,\delta)$-differentially private query release mechanism in the interactive setting, with running time per query $\tilde{O}(m/\alpha^2)$ that is $(\alpha,\beta)$-accurate with respect to the set of all $m$-sparse linear queries $\queryset_m$, with:
$$\alpha = O\left(\frac{\left(\log m\right)^{1/4}\left(\log\frac{4}{\delta}\log\frac{k}{\beta}\right)^{1/2}}{\left(\epsilon\cdot n\right)^{1/2}}\right)$$
\end{theorem}
\begin{proof}
The proof follows by instantiating Theorem \ref{thm:IDCtoRelease} with the SMW IDC algorithm, together with the bound $B(\alpha) = \frac{4(\log s+1)}{\alpha^2}$ proven in Theorem \ref{thm:IDC}, and recalling that $s$ is the smallest integer such that $s/(\log s+1) \geq 4m/\alpha^2$.
\end{proof}

\subsection{Applications to Conjunctions}
In this section, we briefly mention a simple application of this algorithm to the problem of releasing conjunctions with many literals. The algorithm given in this section leads to new results for releasing conjunctions on $d-k$ out of $d$ literals. This complements the recent results of Hardt, Rothblum, and Servedio \cite{HRS11} for releasing conjunctions on $k$ out of $d$ literals. The class of conjunctions are defined over the universe $\cx = \{0,1\}^d$ equal to the $d$-dimensional boolean hypercube.
\begin{definition}
A conjunction is a linear query specified by a subset of variables $S \subseteq [d]$, and defined by the predicate $Q_S:\{0,1\}^d\rightarrow \{0,1\}$ where $Q_S(x) = \prod_{i \in S} x_i$. We say that a conjunction $Q_S$ has $t$ literals if $|S| = t$.
\end{definition}
\begin{remark}
The set of all conjunctions of $d - k$ literals, denoted $C_{d-k}$ is $2^k$ sparse, and of size $|C| \leq d^k$.
\end{remark}
We can release the answers to all queries in $C_{d-k}$ by running the sparse multiplicative weights algorithm on each query. We therefore get the following corollary:
\begin{corollary}
There exists an $(\epsilon,\delta)$-differentially private algorithm in the non-interactive release setting with running time at most
$$\tilde{O}\left(|C_{d-k}|\cdot \frac{2^k}{\alpha^2} \right)= \tilde{O}\left( \frac{(2d)^k}{\alpha^2}\right)$$
that is $(\alpha,\beta)$-accurate for the set of all conjunctions on $d-k$ literals, which requires a database of size only:
$$n \geq \frac{k^{1.5}\log \frac{1}{\delta}\log \frac{d}{\beta}}{\epsilon\alpha^2}$$
\end{corollary}
We note that the running time of this algorithm is comparable to the running time of the algorithm of \cite{HRS11} for releasing all conjunctions of $k$ out of $d$ literals to worst case error (time roughly $\tilde{O}(|C_k|) = \tilde{O}(d^k)$), but requires a database of size only roughly $k^{1.5}\log d$, rather than $d^{\tilde{O}(\sqrt{k})}$ as required by \cite{HRS11}. Of course, conjunctions on $k$ literals are a more natural class than conjunctions on $d-k$ literals, but the results are complimentary.

Moreover, applying the sparse multiplicative weights algorithm in the interactive setting gives polynomially bounded running time per query for conjunctions on $d-k$ literals for any $k = O(\log n)$. Note that this is still a super-polynomially sized class of conjunctions, with $|C_{O(\log n)}| = d^{O(\log n)}$. This is the first interactive query release algorithm that we are aware of that is simultaneously privacy-efficient and computationally-efficient for a super-polynomially sized class of conjunctions (or any other family of queries with super-constant VC-dimension).   
\section{A Non-Interactive Mechanism via Random Projection}
\label{sec:noninteractive}
In this section, we give a non-interactive query release mechanism for sparse queries based on releasing a perturbed random projection of the private database, together with the projection matrix. Note that when viewing the database $\dba$ as a vector, it is an $|\univ|$-dimensional object: $\dba \in \mathbb{R}^{|\univ|}$. A linear projection of $\dba$ into $T$ dimensions is obtained by multiplying it by a $|\univ|\times T$ matrix, which cannot even be represented explicitly if we require algorithms that run in time polynomial in $n = |\dba|$ for $n << |\univ|$. It is therefore essential that we use projection matrices which can be represented concisely using hash functions drawn from limited-independence families.

We will use a limited-independence version of the Johnson-Lindenstrauss lemma presented in \cite{KN11}, first proven by \cite{Ach01, CW09}.
\begin{theorem}[The Johnson-Lindenstrauss Lemma with Limited Independence \cite{Ach01,CW09,KN11}]
\label{thm:JL}
For $d > 0$ an integer and any $0 < \varsigma, \tau < 1/2$, let $A$ be a $T\times d$ random matrix with $\pm 1/\sqrt{T}$ entries that are $r$-wise independent for $T \geq 4\cdot 64^2\varsigma^{-2}\log(1/\tau)$ and $r \geq 2\log(1/\tau)$. Then for any $x \in \mathbb{R}^d$:
$$\Pr_A[|||Ax||_2^2 - ||x||_2^2| \geq \varsigma ||x||_2^2] \leq \tau$$
\end{theorem}
We will use the fact that random projections also preserve pairwise inner products. The following corollary is well known:
\begin{corollary}
\label{cor:JL}
For $d > 0$ an integer and any $0 < \varsigma, \tau < 1/2$, let $A$ be a $T\times d$ random matrix with $\pm 1/\sqrt{T}$ entries that are $r$-wise independent for $T \geq 4\cdot 64^2\varsigma^{-2}\log(1/\tau)$ and $r \geq 2\log(1/\tau)$. Then for any $x, y \in \mathbb{R}^d$:
$$\Pr_A[|\langle(Ax),(Ay)\rangle - \langle x, y\rangle | \geq \frac{\varsigma}{2} (||x||_2^2+||y||_2^2)] \leq 2\tau$$
\end{corollary}
\begin{proof}
Consider the two vectors $u = x+y$ and $v = x-y$. We apply Theorem \ref{thm:JL} to $u$ and $v$. By a union bound, except with probability $2\tau$ we have: $|||A(x+y)||_2^2 - ||x+y||_2^2| \leq \varsigma||x+y||_2^2$ and $|||A(x-y)||_2^2 - ||x-y||_2^2| \leq \varsigma||x-y||_2^2$. Therefore:
\begin{eqnarray*}
\langle (Ax),(Ay)\rangle  &=& \frac{1}{4}\left(\langle A(x+y), A(x+y)\rangle  - \langle A(x-y), A(x-y)\rangle \right) \\
&=& \frac{1}{4}\left(||A(x+y)||_2^2 + ||A(x-y)||_2^2 \right)\\
&\leq& \frac{1}{4}\left((1+\varsigma)||x+y||_2^2 - (1-\varsigma)||x-y||_2^2\right) \\
&=& \langle x, y\rangle  + \frac{\varsigma}{2}\left(||x||_2^2 + ||y||_2^2\right)
\end{eqnarray*}
An identical calculation shows that $\langle(Ax),(Ay)\rangle \geq \langle x, y\rangle  - \frac{\varsigma}{2}\left(||x||_2^2 + ||y||_2^2\right)$, which completes the proof.
\end{proof}

\begin{definition}[Random Projection Data Structure]
The random projection datastructure $\dba_r$ of size $T$ is composed of two parts: we write $\dba_r = (u, f)$.
\begin{enumerate}
\item $u \in \mathbb{R}^T$ is a vector of length $T$.
\item $f:[|\univ|\cdot T]\rightarrow \{-1/\sqrt{T},1/\sqrt{T}\}$ is a hash function implicitly representing a $T \times |\univ|$ projection matrix $A \in \{-1/\sqrt{T},1/\sqrt{T}\}^{T\times |\univ|}$. For any $(i,j) \in T \times |\univ|$, we write $A[i,j]$ for $f(|\univ|\cdot(i-1)+j)$.
\end{enumerate}
To evaluate a linear query $\query$ on a random projection datastructure $\dba_r = (u,f)$ we first project the query and then evaluate the projected query. To project the query we compute a vector $\widehat{\query} \in \mathbb{R}^T$ has follows. For each $i \in [T]$
$$\widehat{\query}[i] = \sum_{x \in \univ : \query(x) > 0}\query[x]\cdot A[i,x]$$
Then we output: $\query(\dba_r) = \frac{1}{n}\langle \widehat{\query}, u\rangle $.
\end{definition}

\begin{algorithm}
\caption{SparseProject takes as input a private database $\dba$ of size $n$, privacy parameters $\epsilon$ and $\delta$, a confidence parameter $\beta$, a sparsity parameter $m$, and the size of the target query class $k$.}
\textbf{SparseProject}($\dba, \epsilon,\delta, \beta, m, k$)
\begin{algorithmic}
\STATE \textbf{Let} $\tau \leftarrow \frac{\beta}{4k}$, $T \leftarrow 4\cdot 64^2\cdot \log\left(\frac{1}{\tau}\right)\left(\frac{m^{3/2}}{2} + \frac{n^4}{2\sqrt{m}}+\sqrt{m}n^2 \right)$,    $\sigma \leftarrow \frac{\epsilon}{\sqrt{8\ln(1/\delta)}}$
\STATE \textbf{Let} $f$ be a randomly chosen hash function from a family of $2\log(k T/2\beta)$-wise independent hash functions mapping $[T\times |\univ|] \rightarrow \{-1/\sqrt{T},1/\sqrt{T}\}$. Write $A[i,j]$ to denote $f(|\univ|\cdot(i-1)+j)$.
\STATE \textbf{Let} $u, \nu \in \mathbb{R}^T$ be a vectors of length $T$.
\FOR{$i = 1$ to $T$}
    \STATE \textbf{Let} $u_i \leftarrow \sum_{x : \dba[x] > 0}\dba[x]\cdot A[i,x]$
    \STATE \textbf{Let} $\nu_i \leftarrow \Lap(1/\sigma)$
\ENDFOR
\STATE \textbf{Output} $\dba_r = (u+\nu, f)$.
\end{algorithmic}
\end{algorithm}

\begin{remark}
There are various ways to select a hash function from a family of $r$-wise independent hash functions mapping $[T\times|\univ|] \rightarrow \{0,1\}$. The simplest, and one that suffices for our purposes, is to select the smallest integer $s$ such that $2^s \geq T \times |\univ|$, and then to let $f$ be a random degree $r$ polynomial in the finite field $\mathbb{GF}[2^s]$. Selecting and representing such a function takes time and space $O(r\cdot s) = O(r(\log |\univ| + \log T))$. $f$ is then an unbiased $r$-wise independent hash function mapping $\mathbb{GF}[2^s] \rightarrow \mathbb{GF}[2^s]$. Taking only the last output bit gives an unbiased $r$-wise independent hash function mapping $[T\times|\univ|]$ to  $\{0,1\}$, as desired.
\end{remark}

\begin{theorem}
SparseProject is $(\epsilon,\delta)$-differentially private.
\end{theorem}
\begin{proof}
For each $i$, write $u_i(\dba) = \sum_{x : \dba[x] > 0}\dba[x]\cdot A[i,x]$. Note that because each entry of $A$ has magnitude $1/\sqrt{T}$, for any database $\dbb$ that is neighboring with $\dba$, $|u_i(\dba) - u_i(\dbb)| \leq 1/\sqrt{T}$. Therefore by Theorem \ref{thm:laplace-privacy}, releasing $u_i + \nu_i$ preserves $(\epsilon/(\sqrt{8T\ln(1/\delta)}), 0)$-differential privacy. We may now apply the composition Theorem \ref{thm:compose} to find that releasing all $T$ coordinates of $u + \nu$ preserves $(\epsilon,\delta)$-differential privacy. Note that $f$ was chosen independently of $\dba$, and releasing it has no privacy cost.
\end{proof}

We first give a high probability bound on the maximum magnitude of any coefficient $\widehat{\query}_i$ of a projected query for any query $\query \in \queryset$. If we were using a random sign matrix for our projection, the following lemma would be a consequence of a simple Chernoff bound, but because we are using only a limited independence family of random variables, we must be more careful.

\begin{lemma}
Let $\queryset$ be a collection of $m$-sparse linear queries of size $|\queryset| = k$,  and $A \in \mathbb{R}^{T\times |\univ|}$ be a matrix with $r$-wise independent entries taking values in $\{-1/\sqrt{T},1/\sqrt{T}\}$, for some even integer $r$. Denote the projection of $\query \in \queryset$ by $A$ by $\widehat{\query} \in \mathbb{R}^T$.  Then except with probability at most $\beta$
$$\max_{\query \in \queryset}\max_{i \in [T]}|\widehat{\query}[i]| \leq \left(\frac{k\cdot T}{2\beta}\right)^{1/r}\cdot \frac{2\sqrt{m r}}{\sqrt{T}}$$
\end{lemma}
\begin{proof}
We follow the approach of Bellare and Rompel \cite{BR94, DPBook}. Recall that for any query $\query$,  $\widehat{\query} \in \mathbb{R}^T$ is defined to be the vector such that $\widehat{\query}[i] = \sum_{x \in \univ : \query(x) > 0}\query[x]\cdot A[i,x]$. Note that each coordinate is dominated by the sum of at most $m$ $r$-wise independent Rademacher random variables (i.e. Bernoulli random variables taking values in $\{-1,1\}$): $\query[i] \leq \frac{1}{\sqrt{T}}\sum_{i=1}^m R_i$, and so it is sufficient to bound this sum. Equivalently, we can write $\query[i] \leq \frac{1}{\sqrt{T}}\left(2\sum_{i=1}^mB_i - m\right)$, where the $B_i$s are $r$-wise independent Bernoulli random variables. Let $B = \sum_{i=1}^m B_i$. By Markov's inequality, we have:
\begin{equation}\
\label{eq:markov}
\Pr\left[|B - \frac{m}{2}| > t\right] = \Pr\left[(B-\frac{m}{2})^r > t^r\right] < \frac{\Ex\left[(B-\frac{m}{2})^r\right]}{t^r}
\end{equation}
Note that because the $B_i$s are $r$-wise independent, we have $\Ex\left[(B-\frac{m}{2})^r\right] = \Ex\left[(\widehat{B}-\frac{m}{2})^r\right]$ where $\widehat{B}$ is the sum of $m$ \emph{truly} independent Bernoulli random variables. We can therefore apply a standard Chernoff bound to control $\widehat{B}$:
\begin{eqnarray*}
\Ex\left[(\widehat{B}-m/2)^r\right] &=& \int_{0}^{\infty}\Pr\left[|\widehat{B}-m/2| > t^{1/r}\right] dt \\
&\le& \int_{0}^{\infty}\exp\left(-\frac{2t^{2/r}}{m}\right) dt \\
&=& \left(\frac{m}{2}\right)^{r/2}\left(\frac{r}{2}\right)! \\
&\leq& e^{1/6r}\sqrt{\pi r}\left(\frac{mr}{4e}\right)^{r/2}
\end{eqnarray*}
where the first inequality follows by a Chernoff bound and the second inequality follows by Stirlings approximation\footnote{The form of Stirlings approximation that we use is: $$k! < e^{1/(12k)}\sqrt{2\pi k}\left(\frac{k}{e}\right)^k$$}. Plugging this in to Equation \ref{eq:markov}, we find:
\begin{equation}
\label{eqn:limited-chernoff}
\Pr\left[|B - \frac{m}{2}| > t\right] < 2\left(\frac{mr}{t^2}\right)^{r/2}
\end{equation}
Recall that $|\widehat{\query}[i]| > c$ if and only if $|B-\frac{m}{2}| > \frac{\sqrt{T}}{2}\cdot c$. Applying Equation \ref{eqn:limited-chernoff} and taking a union bound over all $k$ queries and $T$ indices per query proves the lemma.
\end{proof}
\begin{corollary}
\label{cor:magnitude}
Let $\queryset$ be a collection of $m$-sparse linear queries of size $|\queryset| = k$,  and $A \in \mathbb{R}^{T\times |\univ|}$ be a matrix with $r$-wise independent entries taking values in $\{-1/\sqrt{T},1/\sqrt{T}\}$, for some integer $r > \log\left(\frac{kT}{2\beta}\right)$. Denote the projection of $\query \in \queryset$ by $A$ by $\widehat{\query} \in \mathbb{R}^T$.  Then except with probability at most $\beta$
$$\max_{\query \in \queryset}\max_{i \in [T]}|\widehat{\query}[i]| \leq 4 \cdot \frac{\sqrt{m \log(kT/2\beta)}}{\sqrt{T}}$$
\end{corollary}

We will also make use of a tail bound for sums of Laplace random variables. This bound is likely well known. We use a version proven in \cite{GRU11}.
\begin{lemma}[\cite{GRU11}]
\label{lem:laplaceconcentration}
Suppose that $\{Y_i\}_{i=1}^T$ are i.i.d. $\Lap(b)$ random variables, and scalars $q_i \in [-B,B]$. Define $Y = \sum_{i=1}^Tq_iY_i$. Then:
$$\Pr[|Y| \geq B\alpha] \leq \left\{
                             \begin{array}{ll}
                               \exp\left(-\frac{\alpha^2}{6Tb^2}\right), & \hbox{If $\alpha \leq Tb$;} \\
                               \exp\left(-\frac{\alpha}{6b}\right), & \hbox{If $\alpha > Tb$.}
                             \end{array}
                           \right.$$
\end{lemma}
We can now prove a utility theorem for SparseProject:
\begin{theorem}
For any $0 < \epsilon,\delta < 1$, and any $\beta < 1$, and with respect to any class of $m$-sparse linear queries $\queryset \subset \queryset_m$ of cardinality $|\queryset| \leq k$, SparseProject is $(\alpha,\beta)$-accurate for:
$$\alpha = \tilde{O}\left(\log\left(\frac{k}{\beta}\right)\frac{\sqrt{m\log\left(\frac{1}{\delta}\right)}}{\epsilon n}\right)$$
where the $\tilde{O}$ hides a term logarithmic in $(m + n)$.
\end{theorem}
\begin{proof}
Let $\dba_r = (\hat{u},f)$ be the random-projection data-structure output by SparseQueries, where $\hat{u} = u + \nu$. Consider any fixed query $\query \in \queryset$. Let $\widehat{\query} \in \mathbb{R}^T$ denote the projection of $\query$ by the matrix implicitly defined by $f$. We have:
$$\query(\dba_r) = \frac{1}{n}\langle \widehat{\query}, \hat{u} \rangle = \frac{1}{n}\left(\langle \widehat{\query}, u \rangle + \langle \widehat{\query}, \nu \rangle\right)$$
We will have two sources of error: distortion from the random projection, which we will analyze using the Johnson-Lindenstrauss lemma, and error introduced because of the Laplace noise added for privacy. We will analyze each source separately, starting with the error from the random projection.

Recall that we selected $\tau = \frac{\beta}{4k}$ and $T = 4\cdot 64^2 \varsigma^{-2}\log(1/\tau)$ for $\varsigma = \frac{2\sqrt{m}}{m+n^2}$. Therefore, applying Corollary \ref{cor:JL} together with a union bound over all $k$ queries $\query \in \queryset$, except with probability at most $\beta/2$:
\begin{eqnarray*}
\max_{Q \in \queryset}|\langle \query, \dba \rangle - \langle \widehat{\query}, u \rangle| &\le& \frac{\varsigma}{2}(||\dba||_2^2 + ||\query||_2^2) \\
&\leq& \frac{\varsigma}{2}(n^2 + m) \\
&=& \sqrt{m}
\end{eqnarray*}

We now consider the error introduced by the Laplace noise $\nu$. We first apply Corollary \ref{cor:magnitude} to see that except with probability at most $\beta/4$, we have:
$$\max_{\query \in \queryset}\max_{i \in [T]}|\widehat{\query}[i]| \leq 4 \cdot \frac{\sqrt{m \log(2kT/\beta)}}{\sqrt{T}}$$
Conditioning on this event occurring, we may apply Lemma \ref{lem:laplaceconcentration} with $B = 4 \cdot \frac{\sqrt{m \log(2kT/\beta)}}{\sqrt{T}}$ together with a union bound over all $k$ queries $\query \in \queryset$, to find that except with probability at most $\beta/4$:
\begin{eqnarray*}
\max_{\query \in \queryset} |\langle \widehat{\query}, \nu \rangle| &\leq& 4\sqrt{6 m\left(\frac{1}{\sigma}\right)^2\log\left(\frac{4 k}{\beta}\right)\left(\log\left(\frac{2 k}{\beta}\right) + \log T\right)} \\
&=& \frac{16\sqrt{3}}{\epsilon}\sqrt{m\log\left(\frac{4 k}{\beta}\right)\log\left(\frac{1}{\delta}\right)\left(\log\left(\frac{2 k}{\beta}\right) + \log T\right)} \\
&=& \tilde{O}\left(\log\left(\frac{k}{\beta}\right)\frac{\sqrt{m\log\left(\frac{1}{\delta}\right)}}{\epsilon}\right)
\end{eqnarray*}
where the $\tilde{O}$ is hiding a $\log(T)$ term which is logarithmic in $m$ and $n$.

Finally we can complete the proof. We have shown that except with probability at most $\beta$:
\begin{eqnarray}
\max_{\query \in \queryset}|\query(\dba) - \query(\dba_r)| &=& \frac{1}{n} \max_{\query \in \queryset} | \langle \query, \dba \rangle -  \langle \widehat{\query}, \hat{u} \rangle | \\
&\leq& \frac{1}{n} \max_{\query \in \queryset} \left(|\langle \query, \dba \rangle - \langle \widehat{\query}, u \rangle| + |\langle \widehat{\query}, \nu \rangle|\right) \\
&\leq& \frac{1}{n}\left(\sqrt{m} + \tilde{O}\left(\log\left(\frac{k}{\beta}\right)\frac{\sqrt{m\log\left(\frac{1}{\delta}\right)}}{\epsilon}\right)\right) \\
&=& \tilde{O}\left(\log\left(\frac{k}{\beta}\right)\frac{\sqrt{m\log\left(\frac{1}{\delta}\right)}}{\epsilon n}\right)
\end{eqnarray}
which completes the proof.
\end{proof} 

\subsection{Applications to Conjunctions}
In this section, we again briefly briefly mention a simple application of our non-interactive mechanism to the problem of releasing conjunctions with many literals. This gives the first polynomial time algorithm for non-interactively releasing a super-polynomially sized set of conjunctions.
\begin{definition}
Recall that a conjunction is a linear query specified by a subset of variables $S \subseteq [d]$, and defined by the predicate $Q_S:\{0,1\}^d\rightarrow \{0,1\}$ where $Q_S(x) = \prod_{i \in S} x_i$. We say that a conjunction $Q_S$ has $t$ literals if $|S| = t$.
\end{definition}
\begin{remark}
The set of all conjunctions of $d - k$ literals, denoted $C_{d-k}$ is $2^k$ sparse, and of size $|C_{d-k}| \leq d^k$.
\end{remark}
Sparseproject therefore gives the following corollary:
\begin{corollary}
There exists an $(\epsilon,\delta)$-differentially private algorithm in the non-interactive release setting with polynomially bounded running time, that is $(\alpha, \beta)$-accurate for the class of conjunctions $C_{d-\log n}$ on $d - \log n$ literals for:
$$\alpha = \tilde{O}\left(\left(\log n\log d + \log\frac{1}{\beta}\right)\frac{\sqrt{\log\left(\frac{1}{\delta}\right)}}{\epsilon \sqrt{n}}\right)$$
\end{corollary}
Note that $C_{d-\log n}$ is a super-polynomially sized set of conjunctions. As far as we know, this represents the first algorithm in the non-interactive setting with non-trivial accuracy guarantees for a super-polynomially sized set of conjunctions that also achieves polynomial running time.  

\section{Conclusions and Open Problems}
In this paper, we have given fast interactive and non-interactive algorithms for privately releasing the class of \emph{sparse} queries. Query release algorithms with run-time polynomial in the database size are unfortunately rare, and so a natural question is whether the fast algorithms given here can be leveraged as subroutines in the development of efficient algorithms for other applications. Of course the main question which remains open is to find other classes of queries for which fast data release algorithms exist. Random projections of the database, together with concise representations of the projection matrix seem like a powerful tool. Can they be leveraged in a setting beyond the case of sparse queries, when the norm of the queries are comparable to the norm of the database itself?


\bibliographystyle{alpha}
\bibliography{sparsequeries}

\newcommand{\etalchar}[1]{$^{#1}$}
\begin{thebibliography}{DNR{\etalchar{+}}09}

\bibitem[Ach01]{Ach01}
D.~Achlioptas.
\newblock Database-friendly random projections.
\newblock In {\em Proceedings of the twentieth ACM SIGMOD-SIGACT-SIGART
  symposium on Principles of database systems}, page 281. ACM, 2001.

\bibitem[BLR08]{BLR08}
A.~Blum, K.~Ligett, and A.~Roth.
\newblock A learning theory approach to non-interactive database privacy.
\newblock In {\em Proceedings of the 40th annual ACM symposium on Theory of
  computing}, pages 609--618. ACM, 2008.

\bibitem[Blu90]{Blu90}
A.~Blum.
\newblock Learning boolean functions in an infinite attribute space.
\newblock In {\em Proceedings of the twenty-second annual ACM symposium on
  Theory of computing}, pages 64--72. ACM, 1990.

\bibitem[BR94]{BR94}
M.~Bellare and J.~Rompel.
\newblock Randomness-efficient oblivious sampling.
\newblock In {\em Proceedings of the 35th Annual Symposium on Foundations of
  Computer Science}, pages 276--287. IEEE Computer Society, 1994.

\bibitem[CKKL11]{CKKL11}
M.~Cheraghchi, A.~Klivans, P.~Kothari, and H.K. Lee.
\newblock Submodular functions are noise stable.
\newblock {\em Arxiv preprint arXiv:1106.0518}, 2011.

\bibitem[CW09]{CW09}
K.L. Clarkson and D.P. Woodruff.
\newblock Numerical linear algebra in the streaming model.
\newblock In {\em Proceedings of the 41st annual ACM symposium on Theory of
  computing}, pages 205--214. ACM, 2009.

\bibitem[DMNS06]{DMNS06}
C.~Dwork, F.~McSherry, K.~Nissim, and A.~Smith.
\newblock {Calibrating noise to sensitivity in private data analysis}.
\newblock In {\em Proceedings of the Third Theory of Cryptography Conference
  TCC}, volume 3876 of {\em Lecture Notes in Computer Science}, page 265.
  Springer, 2006.

\bibitem[DNP{\etalchar{+}}10]{DNPRY10}
C.~Dwork, M.~Naor, T.~Pitassi, G.N. Rothblum, and S.~Yekhanin.
\newblock Pan-private streaming algorithms.
\newblock In {\em In Proceedings of ICS}, 2010.

\bibitem[DNR{\etalchar{+}}09]{DNRRV09}
C.~Dwork, M.~Naor, O.~Reingold, G.N. Rothblum, and S.~Vadhan.
\newblock {On the complexity of differentially private data release: efficient
  algorithms and hardness results}.
\newblock In {\em Proceedings of the 41st annual ACM Symposium on the Theory of
  Computing}, pages 381--390. ACM New York, NY, USA, 2009.

\bibitem[DP09]{DPBook}
D.~Dubhashi and A.~Panconesi.
\newblock {\em Concentration of measure for the analysis of randomized
  algorithms}.
\newblock Cambridge University Press, 2009.

\bibitem[DRV10]{DRV10}
C.~Dwork, G.N. Rothblum, and S.~Vadhan.
\newblock Boosting and differential privacy.
\newblock In {\em Proceedings of the 51st Annua IEEEl Symposium on Foundations
  of Computer Science}, pages 51--60. IEEE, 2010.

\bibitem[GHRU11]{GHRU11}
A.~Gupta, M.~Hardt, A.~Roth, and J.~Ullman.
\newblock {Privately Releasing Conjunctions and the Statistical Query Barrier}.
\newblock In {\em Proceedings of the 43rd annual ACM Symposium on the Theory of
  Computing}. ACM New York, NY, USA, 2011.

\bibitem[GRU11]{GRU11}
A.~Gupta, A.~Roth, and J.~Ullman.
\newblock Iterative constructions and private data release.
\newblock {\em Arxiv preprint arXiv:1107.3731}, 2011.

\bibitem[HR10]{HR10}
M.~Hardt and G.N. Rothblum.
\newblock A multiplicative weights mechanism for privacy-preserving data
  analysis.
\newblock In {\em 51st Annual IEEE Symposium on Foundations of Computer
  Science}, pages 61--70. IEEE, 2010.

\bibitem[HRS11]{HRS11}
M.~Hardt, G.N. Rothblum, and R.A. Servedio.
\newblock Private data release via learning thresholds.
\newblock {\em Arxiv preprint arXiv:1107.2444}, 2011.

\bibitem[HT10]{HT10}
M.~Hardt and K.~Talwar.
\newblock {On the Geometry of Differential Privacy}.
\newblock In {\em The 42nd ACM Symposium on the Theory of Computing, 2010.
  STOC'10}, 2010.

\bibitem[KN10]{KN11}
D.M. Kane and J.~Nelson.
\newblock A derandomized sparse johnson-lindenstrauss transform.
\newblock {\em Arxiv preprint arXiv:1006.3585}, 2010.

\bibitem[RR10]{RR10}
A.~Roth and T.~Roughgarden.
\newblock {Interactive Privacy via the Median Mechanism}.
\newblock In {\em The 42nd ACM Symposium on the Theory of Computing, 2010.
  STOC'10}, 2010.

\bibitem[UV11]{UV11}
Jonathan Ullman and Salil~P. Vadhan.
\newblock {PCP}s and the hardness of generating private synthetic data.
\newblock In Yuval Ishai, editor, {\em TCC}, volume 6597 of {\em Lecture Notes
  in Computer Science}, pages 400--416. Springer, 2011.

\end{thebibliography}

\end{document}